\documentclass[11pt]{article}
\usepackage[
  bookmarks=true,
  bookmarksnumbered=true,
  bookmarksopen=true,
  pdfborder={0 0 0},
  breaklinks=true,
  colorlinks=true,
  linkcolor=black,
  citecolor=black,
  filecolor=black,
  urlcolor=black,
]{hyperref}

\usepackage{thmtools}
\usepackage{thm-restate}

\usepackage[all=normal,paragraphs,wordspacing,tracking]{savetrees}
\usepackage[margin=1in]{geometry}
\usepackage{setspace}
\usepackage{amsthm}
\usepackage{subfigure}
\usepackage{amsmath}
\usepackage{amsfonts}
\usepackage{amssymb}
\usepackage{graphicx}
\usepackage{pdfpages}
\usepackage{stmaryrd}
\usepackage{enumitem}
\usepackage[capitalise]{cleveref}
\usepackage{nicefrac}
\usepackage{verbatim}
\usepackage[normalem]{ulem}
\usepackage{verbatimbox} 
\usepackage{bm}
\usepackage{marginnote}
\usepackage{tikz}

\newtheorem{theorem}{Theorem}
\newtheorem{definition}{Definition}
\newtheorem{lemma}{Lemma}

\newtheorem*{notation*}{Notation}

\newenvironment{sketch}{\noindent{\it Sketch of Proof.}}{\hfill $\Box$}

\newcommand{\gnew}[1]{{\color{blue}{#1}}}

\newcommand{\RP}{R_{\cal P}}

\newcommand{\decision}{\mathsf{dec}}

\newcommand{\dec}{\mbox{decide}}
\newcommand{\err}{$\overline{\mbox{\textit{err}}}$}
\newcommand{\help}{\mbox{\textit{!`help!}}}
\newcommand{\Proc}{\mathbb{P}}

\newcommand{\vcom}{\tilde{v}}
\newcommand{\deccom}{\widetilde{dec}}

\newcommand{\defemph}[1]{\textbf{\textit{#1}}}

\newcommand{\modelt}{\gamma^t}

\renewcommand{\cal}{\mathcal}

\newcommand{\BB}{\mathtt{Base.Protocol}}
\newcommand{\CC}{\mathtt{CC}}

\usepackage[ruled,linesnumbered,vlined]{algorithm2e}

\renewcommand{\cref}{\Cref}
\newcommand{\ind}[2]{\approx_{#1}^{#2}}

\newcommand{\Prot}{{\cal P}}

\newcommand{\SCR}{{\textsf{scr}}} 
\newcommand{\SCRo}[1]{{\textsf{scr\;\!\!}(#1)}} 
\newcommand{\Gfact}[2]{{\bar\varphi_{#1}^{#2}}} 
\newcommand{\alldec}{\mathtt{all\_decided}}
\newcommand{\maj}{\mbox{$MAJ$}}
\newcommand{\plural}{\mbox{$PLUR$}}

\newcommand{\Lone}{L_1}
\newcommand{\Ltwo}{L_2}
\newcommand{\Lthree}{L_3}

\newcommand{\udash}[1]{%
    \tikz[baseline=(todotted.base)]{
        \node[inner sep=1pt,outer sep=0pt] (todotted) {#1};
        \draw[dashed] (todotted.south west) -- (todotted.south east);
    }%
}
\newcounter{linecounter} 
\newcommand{\linenumbering}{\ifthenelse{\value{linecounter}<10} 
	{(0\arabic{linecounter})}{(\arabic{linecounter})}} 
\renewcommand{\line}[1]{\refstepcounter{linecounter}\label{#1}\linenumbering} 
 
\renewcommand{\thelinecounter}{\ifnum \value{linecounter} >  
	9\else 0\fi \arabic{linecounter}} 

\newlength {\squarewidth}

\title{
Byzantine Consensus in the Common Case}
\author{Guy Goren\\
	Department of Electrical Engineering, Technion\\
	\texttt{sgoren@campus.technion.ac.il}
	\and 
	Yoram Moses\\
	Department of Electrical Engineering, Technion\\
	\texttt{moses@ee.technion.ac.il}  
}
\date{}
\begin{document}
	\maketitle
\begin{abstract}
Modular methods to transform Byzantine consensus protocols into ones that are fast and communication efficient in the common cases are presented. 
Small and short protocol segments called layers are custom designed to optimize performance  in the common case. When composed with a Byzantine consensus protocol of choice, 
they allow considerable control over the tradeoff in the combined protocol's behavior in the presence  of failures and its performance in their absence.
When runs are failure free in the common case, 
the resulting protocols decide in two rounds and require $2nt$ bits of communication.
For  the common case assumption that all processors propose~1 and no failures occur, 
we show a transformation 
in which decisions are made in one round, and no bits of communication are exchanged. 
 The resulting protocols achieve better common-case complexity than all existing Byzantine consensus protocols. Finally, in the rare instances in which the common case does not occur, a small cost is added to the complexity of the original consensus protocol being transformed.  
 The key ingredient of these layers that allows both time and communication efficiency in the common case is the use of {\em silent confirmation rounds}, which are rounds where considerable relevant information can be obtained in the absence of any communication whatsoever. 
\end{abstract}

	\bigskip\bigskip
$~$\hfill\begin{tabular}{ll}
	\textit{``I am prepared for the worst, but hope for the best"} \\[.5ex]
	&\hspace{-1cm}{Benjamin Disraeli}
\end{tabular}

\medskip\vfill
\noindent
\vfill

\thispagestyle{empty}
\newpage
\setcounter{page}{1}

\section{Introduction}
\label{sec:Intro}

Reaching agreement on values is a fundamental problem in distributed systems. While voting is a natural mechanism for this purpose, it is not implementable when failures are possible. 
In their seminal paper Pease, Shostak and Lamport~\cite{PSL} defined the Byzantine consensus problem (originally called Interactive Consistency) in 1980. 
Broadly speaking, Byzantine consensus considers the problem of reaching agreement among a group of~$n$ parties,  up to~$t$ of which can be Byzantine faults and deviate from the protocol arbitrarily. 
Pease, Shostak and Lamport presented a protocol that solves the problem in~$t+1$ rounds whenever~$n>3t$, and proved that no solution for~$n\le3t$ exists~\cite{PSL,LSP}. 
Fischer and Lynch later showed that~$t+1$ rounds are necessary in the worst-case run of any Byzantine consensus protocol~\cite{FL81}. 
In the original solution of~\cite{PSL,LSP}, processes never decide before the end of~$t+1$ rounds. Moreover, each process sends an exponential number of bits of information (and performs an exponential amount of computation)  in every execution. The authors leave as an open  problem the design of more efficient solutions to Byzantine consensus, and the quest for efficient solutions to this problem has received a great deal of attention over the last four decades. For a recent partial survey, see~\cite{AD15}. Currently available solutions present a variety of options for trading off the time complexity of Byzantine consensus protocols, their communication costs, and the fault-tolerance ratio (the ratio between $n$ and~$t$) that they provide.


In the early days of the subject, the Byzantine failure model was considered somewhat esoteric and unrealistic, and so it was not considered cost-effective to tolerate such failures in most practical systems. 
This has changed over the years, however. The critical role played by distributed systems nowadays and the increasing abundance of  cyber threats have made Byzantine-fault tolerance practically relevant \cite{CL99PBFT, KADC07}. 

Whereas the damage that malicious cyber attacks may cause is considerable and tolerating Byzantine faults is important, such faults are typically rare.
Therefore, a system that incurs high costs only in the rare cases when these faults occur and is cheaper in all other cases, is desirable from a practical point of view.
Dolev et al.~\cite{DRS}, introduced ``early stopping'' solutions that are adaptive to the number of actual failures~$f$ in an execution. 
Early stopping protocols often decide after~$\min\{t+1,f+2\}$ rounds \cite{AD15,DRS,MW88,PR04}.
In practice, the case of $f=0$ in which no failures occur is typically much more common than all other cases.
We contend that optimizing Byzantine consensus for $f=0$ rather than for general $f<t$ may be worthwhile.
Indeed, well-known practical solutions to the state replication problem have focused on optimizing performance for runs in which no failure occurs (see, e.g., \cite{CL99PBFT,KADC07}). 
In this paper we investigate the design of Byzantine consensus protocols that are especially efficient in this important common case.
We show that optimizing for $f=0$ instead of for any~$f$ yields simpler and more efficient solutions.

Protocols that are highly efficient in the common case can make a significant contribution to the effective operations of a distributed system \cite{CL99PBFT,ADDNR17,MA06,BDFG03,K02,L06,HH93,KADC07,GM18,GW17}. 
If failures are rare, then the system's performance in the common, failure-free, case is much more practically relevant than its performance when failures occur. 
Nevertheless, a protocol designer may prefer the tradeoffs offered by a particular solution in the presence of failures. Our goal is to provide tools that will guarantee excellent behavior in the common case, and will allow the protocol designer to 
switch to her favorite protocol when failures occur.
Intuitively, this is achieved as follows. 
Given an arbitrarily chosen binary consensus protocol, which we refer to as the {\em basic} protocol, we prepend a layer consisting of a very small number of rounds to the basic protocol. 
The combined protocol executes this layer, and reaches consensus quickly and efficiently in the common case. If it has not reached consensus, execution transfers to the basic protocol, which proceeds to achieve consensus as usual. 
Even in the rare uncommon case, the prepended layer does not increase the asymptotic communication cost of the basic protocol, and it adds only a small constant number of rounds to the running time. 
Turpin and Coan have similarly designed an initial protocol layer that transforms a binary consensus protocol into a multi-valued protocol~\cite{TurpinCoan84}.
\\[1ex]

\noindent The main contributions of this paper are:
\begin{itemize}
\item We present an optimal optimizing layer~$\Lone$ for the stronger common case assumption that all processes propose~1 and no failures occur. In the common case, it decides after one round, and uses {\em no communication whatsoever}. 
The use of~$\Lone$ makes consensus essentially free in this common case. In the presence of failures, $\Lone$ adds one round and a total of $n^2$ bits of communication to the basic protocol. 
\item For the common case in which no failures occur, we present two optimizing layers. One of them, denoted~$\Ltwo$, decides in (optimal) 2 rounds and uses $2nt$ bits in the common case. All previous consensus protocols that decide in~2 rounds when $f=0$ require $\Omega(n^3)$ bits (cf.~\cite{BGP92,AD15}).
The layer~$\Ltwo$ adds 3 rounds and less than $4n^2$ bits to the basic protocol when failures occur.  
The other layer, $\Lthree$,  costs one more round but it uses half the communication of~$\Ltwo$ in the common case. The communication cost of~$\Lthree$ in the common case is 24 times better than that of the best protocol in the literature (\cite{BGP92}), and is within a factor of~4 from the lower bound of $nt/4$ (\cite{DR85,HH93}). 

\item Whereas decisions in binary consensus protocols are often biased in favor of 
a predetermined value or in favor values of particular processes, our layers decide on the majority value in the common case. They convert consensus into a fair vote in all but relatively rare cases.
%
\item A key ingredient that improves the communication-efficiency of the layers are 
rounds that, in the failure-free case, convey useful global information without communication. We call them {\em silent confirmation rounds}. 
We formalize silent confirmation rounds, and present a theorem that demonstrates how silent confirmation rounds, the null-messages of \cite{lamport1984using} and the silent choirs of \cite{GM18} all emerge from the same principles. 
\end{itemize}

Using one of our layers, the protocol designer is free to choose her favorite base protocol. 
Typical considerations for such a choice may involve the desired tradeoff between the time, communication costs and resilience of the base protocol in the presence of failures. Here are some such tradeoffs provided by existing binary consensus protocols for the Byzantine failure model. 
The coordinated traversal protocol of \cite{MW88} is an early stopping protocol that sends $O(n^3)$ bits per round and tolerates $t<n/8$ Byzantine failures. Abraham and Dolev present an early-stopping protocol that tolerates $t<n/3$ faults, and uses polynomial communication (for an unstated large polynomial, possibly $O(n^8)$ or $O(n^9)$ bits). 	Kowalski and Most\'efaoui present a protocol for $t<n/3$ that decides in~$t+1$ rounds and uses~$O(n^3\log n)$ bits, except that executing it requires exponential computation.
Deciding in more than $t+1$ rounds has allowed solutions with better communication complexity for $t<n/3$. A protocol in~\cite{BGP92bit} achieves~$14n^2$ bit complexity with $t+o(t)$ rounds.
	An adaptation of a protocol by Hadzilacos and Halpern in~\cite{HH93} can be made to cost only~$n^2(t+1)/4$ bits and decide after 3 rounds in failure-free executions.
	The Early Stopping Phase King and Queen protocols of Berman, Garay and Perry in~\cite{BGP92}, focused on reducing communication complexity and achieved outstanding results of~$n^2$ bits per round complexity. The Phase King protocol has optimal resilience and decides after~$\min\{4(f+2),4(t+1)\}$ rounds, while Phase Queen is faster and decides after~$\min\{2(f+2),2(t+1)\}$ rounds, but it has a reduced resilience of~$t<t/4$.
	In all cases, adding a cost of $O(n^2)$ bits by running any of our layers as a preliminary stage does not increase the asymptotic communication complexity. 


The remainder of the paper is organized as follows. The next section formally defines our system model. In \cref{sec:ByzConsProtocols} we present common-case optimizers for Byzantine consensus. \cref{sec:1round,sec:prot2,sec:prot3} describe our three main layers and discusses their properties. 
\cref{sec:SCR} discusses and formalizes silent communication rounds, which are a key element in reducing communication costs in our solutions. 
Finally, concluding remarks are provided in \cref{sec:discussion}. Proofs of all statements appear either in the main text or in the Appendix. 



\section{Model and Preliminaries}
\label{sec:Model}
In the consensus problem, each process~$i$ starts with some initial value $v_i\in V$, and all correct processes need to reach a common decision.
All runs of a consensus protocol are required to satisfy the following conditions: 
\\[1.2ex]
\noindent
\underline{
	{\sc Consensus:}}
\begin{itemize}[itemsep=.5pt]
	\item[]{\bf Decision:}\quad Every correct process must eventually decide, 
	\item[]{\bf Agreement:}\quad All correct processes make the same decision, and 
	\item[]{\bf Validity}:\quad If all correct processes have the same initial value, then all correct processes decide on this value.\\[.3ex]
\end{itemize}
\vspace{-2.5ex}
\noindent
When $V=\{0,1\}$ the problem is called {\em binary} consensus, and when  $|V|>2$ we refer to it as {\em multi-valued} consensus.
A Byzantine consensus protocol is a consensus protocol that can tolerate up to~$t$ byzantine failures per run.

\subsection{Model of Computation}
\label{sec:modelOfComputation}
We consider the standard synchronous message-passing model with Byzantine failures. We assume a set $\Proc=\{0,1,\ldots,n-1\}$ of $n> 2$ processes.
%
Each pair of processes is connected by a two-way communication link, and for each message the receiver knows the identity of the sender.
All processes share a discrete global clock that starts at time~$0$ and  advances by increments of one.
Communication in the system proceeds in a sequence of \emph{rounds}, with round~$m+1$ taking place between time~$m$ and time~$m+1$, for $m\ge 0$. A message sent at time~$m$ (i.e., in round~$m+1$) from a process~$i$ to~$j$ will reach~$j$ by time~$m+1$, i.e., the end of round~$m+1$. 
In every round, each process performs local computations, sends a set of messages to other processes, and finally receives messages sent to it by other processes during the same round.

At any given time~$m\ge 0$, a process is in a well-defined  \defemph{local state}.   For simplicity, we assume that the local state of a process~$i$  consists of its initial value~$v_i$, the current time~$m$, and the sequence of the actions that~$i$ has performed (including the messages it has sent and received) up to that time. In particular, its local state at time~0 has the form~$(v_i,0,\{ \} )$.
A \defemph{protocol} describes what messages a process sends and what decisions it takes, as a deterministic function of its local state. 
Correct (as appose to faulty) process follows the protocol's instructions.
A faulty process, however, behaves arbitrarily and is not restricted to follow the protocol.
In particular, it can act maliciously and send bogus messages in an attempt to sabotage the correct operation of the system.
In a given execution a process is either correct or faulty, it cannot be both.

We will consider the design of protocols that are required to withstand up to~$t$ failures.
Thus, given $1\le t<n$, we denote by $\modelt$ the model described above in which it is guaranteed that no more than~$t$ processes are faulty in any given run.
We assume that a protocol~$\Prot$ has access to the values of~$n$ and~$t$, typically passed to~$\Prot$ as parameters.

A \defemph{run} is a description of a (possibly infinite) execution of the system.
We call a set of runs~$R$ a \defemph{system}.
We will be interested in systems of the form $R_{\Prot}=R({\Prot},\modelt )$ consisting of all runs of a given protocol~$\Prot$ in which no more than~$t$ processes are faulty.
Observe that a protocol~$\Prot$ solves Byzantine consensus in the model~$\modelt$ if and only if every run of~$R_{\Prot}$  satisfies the Decision, Agreement and Validity conditions described above. 
Given a run~$r$ and a time~$m$, we denote the local state of process~$i$ at time~$m$ in run~$r$ by $r_i(m)$.
Notice that a process~$i$ can be in the same local state in different runs of the same protocol. Since the current time~$m$ is represented in the local state $r_i(m)$, however,  $r(m)=r'(m')$ can hold only if $m=m'$.

\subsection{Indistinguishability and Knowledge}
We shall say that two runs~$r$ and~$r'$ are {\em indistinguishable} to process~$i$ at time~$m$ if $r_i(m)=r'_i(m)$. We denote this by $r\ind{i}{m}r'$. Notice that since we assume that correct processes follow deterministic protocols, if $r\ind{i}{m}r'$ then a correct process~$i$ is guaranteed to perform the same actions at time $m$ in both~$r$ and~$r'$. 
Problem specifications  typically impose restrictions on actions, based on properties of the run. Moreover, since the actions that a correct process performs are a function of its local state,  the restrictions can depend on properties of other runs as well. 

For example, the Agreement condition implies that a correct process~$i$ cannot decide on $v$ at time~$m$ in a run~$r$ if there is an indistinguishable run~$r'\ind{i}{m}r$ in which some correct process decides on $u\ne v$. Similarly, by the Validity condition a correct process~$i$ cannot decide on $v$ if there is a run~$r'$ that is indistinguishable from~$r$ (to~$i$ at time~$m$) in which all correct processes have the same initial value~$u\ne v$.
These examples illustrate how indistinguishability can inhibit actions --- performing an action can be prohibited because of what may be true at indistinguishable runs. 

Rather than considering when actions are prohibited, we can choose to consider what is required in order for an action to be allowed by the specification. To this end, we can view the Agreement condition as implying that a correct process~$i$ is allowed to decide on $v$ at time~$m$ in~$r$ only if in every run~$r'\ind{i}{m}r$ there is no correct process that decides otherwise. 

This is much stronger than stating that no correct process decides otherwise in the run~$r$ itself, of course. Roughly speaking, the stronger statement is true because at time~$m$ process~$i$ cannot tell whether it is in~$r$ or in any of the runs $r'\ind{i}{m}r$. When this condition holds, we say that~$i$ {\em knows} that all values are~1. Generally, it will be convenient to define the dual of indistinguishability, i.e., what is true at all indistinguishable runs, as what the process knows. More formally, following in the spirit of \cite{HM1,FHMV}, we proceed to define knowledge in our distributed systems as follows.%
\footnote{We introduce just enough of the theory of knowledge to support the analysis in this paper. For more details, see \cite{FHMV}.}


\begin{definition}[Knowledge]
	\label{def:know}
	Fix a system $R$, a run $r\in R$, a process~$i$ and a fact~$\varphi$.  
	We say that  $K_i\varphi$ (which we read as ``process~$i$ \defemph{knows}~$\varphi$'') holds at time~$m$ in~$r$ iff~$\varphi$ is true of all runs~$r'\in R$ such that $r'\ind{i}{m}r$.
\end{definition}

We shall focus on knowledge of facts that can be extracted from processes local states, their conjunctions and disjunctions. We call a fact that is implied by~$i$'s local state an {\em $i$-local fact}, and denote them by $\varphi_i$, $\psi$, etc.
In particular, facts such as ``$v_i=x$'' ($i$'s initial value is~$x$), {\em ``$i$~received message~$\mu$ from~$j$''} (in the current run), and {\em ``$i$ has decided~$v$''} (in the current run), are all examples of $i$-local facts.
We use Boolean operators such as $\neg$ (Not), $\wedge$ (And), and~$\vee$ (Or) freely in the sequel.

Notice that knowledge is defined with respect to a given system~$R$. Often, the system is clear from context and is not stated explicitly. 
\cref{def:know} immediately implies the so-called {\em Knowledge property}: If $K_i\varphi$ holds at (any) time~$m$ in~$r$, then $r$ satisfies $\varphi$. 


\section{Common-case Optimizers}
\label{sec:ByzConsProtocols}
In this section we describe methods that transform Byzantine consensus protocols into ones that are fast and communication efficient in the common case.
We use a modular approach similar to that of~\cite{TurpinCoan84}. We prepend a few preliminary rounds to an existing consensus protocol that the designer can choose, called the {\em base protocol}.  In the common case, the base protocol is never used, and consensus is achieved with relative ease. Otherwise, the computation transfers to the base protocol after a small number of rounds and relatively efficient communication. 
The behavior of the composed protocol in the rare uncommon case is determined by that of the base protocol.

Throughout the paper we use the following notations. Given a protocol~$\Prot$ and a layer~$L$, we denote the composition of~$L$ and~$\Prot$ by $\CC=L\odot\Prot$. 
In figures depicting a layer we use a \udash{dashed underline} to mark a line that is the only 
command  in its round that the executing process will perform in the common case. We depict the transition to execute a base protocol by painting a box around the base protocol as in line 07 of \cref{fig:bestCase}. 

\subsection{1 Round}\label{sec:1round}

The first layer, which we denote by~$\Lone$, is optimized for the more specific {\em unanimous common case} assumption, where in the common case all processes propose~1 and no failures occur.
Optimizing for the unanimous common case is typical in solutions to atomic commitment in the database literature~\cite{BHG87,DS83}, and is also apparent in the weak Byzantine generals problem of~\cite{L83}. 
In the common case, the layer $\Lone$ ensures fast decisions -- after one round -- and uses~0 bits of communication.
This provides fault-tolerance with negligible costs in the common case.

Roughly speaking, $\Lone$ is structured as follows. In the first round, a process~$i$ sends a one-bit message to all processes if $v_i=0$, 
and remains silent otherwise. A process that receives no message in the first round is informed that all correct processes started with~1. 
A process decides~1 and halts if it receives no message in the first round. 
Otherwise it will participate in the base protocol. It can decide~1 if it receives fewer than $t+1$ messages in the first round. The base protocol will determine the decision value in all other cases. 
In the common case, $\Lone$ decides at the end of the first round, and sends no messages. In any case, $\Lone$ only affects the first round and sends at most $n^2$ bits.

\begin{figure}[h] 
	\centering{ 
		\fbox{ 
			\begin{minipage}[t]{150mm} 
				\footnotesize 
				\renewcommand{\baselinestretch}{2.5} 
				\setcounter{linecounter}{00}
				\begin{tabbing} 
					aaaaaaaa\=aa\=aaaaaaa\=\kill  
					{\bf time 0} \\					
\line{AA01} \>		\udash{{\bf if} $v_i = 1$ {\bf then} be silent} \\[1ex]
\line{AA02} \>		{\bf else} send `\err' to all \quad \% \gnew{$v_i=0$} \%
					\\[1ex]
					{\bf time 1} \textit{and beyond}\\
\line{AA11} \>		\udash{{\bf if} received no `\err' {\bf then} \dec(1); halt}\\[1ex]
\line{AA12} \>		{\bf if} received at most~$t$ `\err' messages {\bf then} \dec(1)\\
\line{AA13} \>		{\bf if} received at most~$2t$ `\err' messages {\bf then} $est_i\gets 1$\\
\line{AA14} \>		{\bf else} $est_i\gets v_i$\\
\line{AA15} \>		$\decision\gets \mbox{\fbox{$\BB(est_i)$}}$\\
\line{AA16} \>		{\bf if} undecided after time~1 {\bf then} \dec $(\decision)$
				\end{tabbing} 
				\normalsize 
			\end{minipage} 
		} 
		\caption{$\Lone$ -- A layer for process~$i$, optimized for the unanimous common case.}
		\label{fig:bestCase} 
	} 
\end{figure}

\begin{restatable}{theorem}{oneround}
\label{thm:1round}
Let $k\ge3$ and let~$\BB$ be a binary consensus protocol for $n>kt$ processes. Then $\CC1=\Lone\odot\BB$ is a binary consensus protocol in which 
\vspace{-2mm}
\begin{enumerate}
\setlength\itemsep{-1mm}

\item In the unanimous common case,  decisions occur after 1 round and no messages are sent, while  
\item In the other cases, at most $n^2$ bits are sent and 1 round elapses before reverting to the base protocol. 
\end{enumerate}
\end{restatable}

\begin{proof}
	The proof of \cref{thm:1round} appears in~\cref{sec:proofs}, as do those of all statements that are not proved in the main body of the paper. 
\end{proof}
\subsection{2 Rounds}
\label{sec:prot2}
The unanimous common case for which~$\Lone$ was designed is a very strong assumption. In many settings it is natural to consider the common case to be a failure-free execution.
The layer~$\Ltwo$ optimizes binary consensus solutions for failure-free executions.%
\footnote{Due to space limitations the full pseudocode of~$\Ltwo$ is presented in~\cref{fig:SanhBinByzCons-protocol} in~\cref{sec:pseude}.}

Roughly speaking, $\Ltwo$ works as follows:\hspace{0.2cm}
A large committee consisting of~$2t+1$ processes is defined \textit{a priori}  (we call it the {\em greater Sanhedrin}, or Sanhedrin, for short). In the first round, all processes inform the greater Sanhedrin of their initial values.
Each member of the Sanhedrin then computes a majority of the votes it received and sends a recommended value to all processes accordingly. 
If no failures occur, the recommendations of the Sanhedrin are all the same, and a process that receives an identical recommendation from everyone decides on this recommendation.
In the common case execution everyone can decide  after two rounds.
In the third round, a process that has decided remains quiet, while a process that received conflicting recommendations asks for help.
In the common case this round serves as a {\em silent confirmation round} for the fact ``all correct processes have decided", a concept that will be formalized and explained in~\cref{sec:SCR}.
A process that has decided at time~2 and receives no help request by time~3 can halt since it knows that all correct processes have decided.
Otherwise it will participate in the base protocol to assist processes who may have not decided yet. 

Observe that in the common case all members of the Sanhedrin receive the same values in the first round. In $\Ltwo$ such a member recommends the majority value among the ones it received in the first round. Hence, in a failure-free round, everyone will decide on the majority value. 

In order to improve the performance of~$\Ltwo$ in its worst-case common case executions, we use a technique due to \cite{amdur1992message}:
In each of the first two rounds, in order to broadcast a binary value~$b\in\{0,1\}$ to a set of processes, a process~$i$ will inform half
of the processes about~$b$ by sending them the value, and the other half by sending them no message. (See lines 02 and 03 in the first round, and lines 07 and 08 in the second round.)
This halves the number of bits required, but results in a slightly more cumbersome algorithm.  
In~\cref{fig:SanhBinByzCons-protocol} recipients are partitioned according to parity of their IDs (Even or Odd), but any other balanced division works equally well.
Layer~$\Ltwo$ uses the majority voting function $MAJ(\cdot): \{0,1\}^n\rightarrow \{0,1\}$, defined as
\[
MAJ\triangleq \left\{
\begin{tabular}{ l l }
1 & \mbox{if at least~$\frac{n}{2}$ votes are~1}\\
0 & \mbox{otherwise} \\
\end{tabular}
\right\}.
\]

\begin{restatable}{theorem}{tworounds}
	\label{thm:2rounds}
	Let $k\ge3$ and let~$\BB$ be a binary consensus protocol for $n>kt$ processes. Then $\CC2=\Ltwo\odot\BB$ is a binary consensus protocol in which 
	\vspace{-2mm}
	\begin{enumerate}
		\setlength\itemsep{-1mm}
		\item In failure-free runs,  decisions occur after 2 rounds and at most $2n(t+1)$ bits are communicated, while
		\item When failures occur, at most $2n(t+1)+n^2$ bits are sent and 3 rounds elapse before reverting to the base protocol. 
		\item In the common case, $\CC2$ decides on the majority value. 
	\end{enumerate}
\end{restatable}
\subsection{3 Rounds}
\label{sec:prot3}
	$\Ltwo$ sends approximately~$2nt$ bits in the common case. 
	In this section we present $\Lthree$ that cuts the communication costs in half to~$nt$ bits, by including one additional round.%
	\footnote{Again, the full pseudocode of $\Lthree$ is presented in \cref{fig:binByzCons-protocol} in \cref{sec:pseude}.}
	The ideas underlying the design of the 3-round layer $\Lthree$ are similar to those of~$\Ltwo$.
	A committee of~$t+1$ processes is \textit{a priori} defined (this time we call it the {\em smaller Council}). 
	In the first round all processes inform the council of their initial values, half in silence and half by messages. 
	Each member of the council then calculates a majority on the votes and sends a recommendation to all processes accordingly. (Again, to half by silence and to the other half by messages.)
	If no failures occur, the members of the council make identical  recommendations.	
	This time, in the third round, a process that receives identical recommendations does not decide but uses the recommendation as its estimation and remains quiet. A process that receives conflicting recommendations sets its estimation to be its initial proposal and sends a message indicating that an error occurred.
	A process for which the third round appears silent discovers that all correct processes also received identical recommendations, and 
	since one of the committee members must be correct, the recommendations they received are the same as the one it received. 
	Consequently, a process that receives no messages in the third round, decides at time~3 on its estimation.	
	In the common case execution no messages are sent in round~3, and every process decides at time~3. 
	Round~4 is dedicated to obtaining the knowledge that ``all correct processes have decided".	Thus, at time~3, a process that has decided remains quiet, while an undecided process asks for help.	
	If no help request arrives by time~4 a process halts, otherwise it turns to the base protocol.
	In the common case all processes decide at time~3, remain quiet in round~4, and halt at time~4.

\begin{restatable}{theorem}{threerounds}
	\label{thm:3rounds}
	Let $k\ge3$ and let~$\BB$ be a binary consensus protocol for $n>kt$ processes. Then $\CC3=\Lthree\odot\BB$ is a binary consensus protocol in which
\vspace{-2mm}
\begin{enumerate}
	\setlength\itemsep{-1mm}
	\item In  failure-free runs,  decisions occur after 3 rounds and at most $n(t+1.5)$ bits are communicated, while  
	\item When failures occur, at most $n(t+1.5)+2n^2$ bits are sent and 4 rounds elapse before reverting to the base protocol. 
	\item In the common case, $\CC3$ decides on the majority value. 
\end{enumerate}
\end{restatable}

The communication cost of the 3-round layer $\Lthree$ is 4 times the best-known lower bound of $\Omega(nt/4)$ bits for this case from \cite{DR85,HH93}.
The previously best-known communication behavior is by the Early Stopping Phase King protocol of \cite{BGP92}, which requires up to $8n^2$ bits, and takes up to~8 rounds to decide in the common case. 
$\Lthree$ achieves a 24-fold improvement in bit complexity, while also reducing the decision time (from 8 to 3 rounds).
In addition, the added costs in the uncommon cases are only 4 rounds and a negligible amount of bits.
Thereby, this layer offers practical systems considerable cost reductions for a minor overhead.

\subsection{Behavior in Uncommon Cases}
\label{sec:unstable consensus}
The layers we have introduced all guarantee that consensus is obtained in the common case without the base protocol ever being called into action. In the uncommon case, if any of the correct processes reverts to the base protocol 
in order to determine its decision value, it first alerts all correct processes, and they all participate in the execution of the base protocol. There is a third possibility, in which all correct processes have decided, but a malicious process falsely alerts some of the correct processes. This can initiate an execution of the base protocol with fewer than $n-t$ correct participants. 
For the purposes of our discussion in this section, we refer to these as {\em redundant executions}.
Since the correctness of the base protocol may rely on the existence of sufficiently many correct participants, this execution might, in general, fail to satisfy the conditions for consensus. 
As this can only happen if all correct processes have already decided before entering the base protocol, it will not affect the correctness of our solution.%
\footnote{Indeed, this is why such executions of the base protocol are called {\em redundant}.}
 It can, however, affect the time and communication costs in redundant executions.


In most popular consensus protocols in the literature, redundant executions,  in which a subset of the processes are initially crashed and  at most~$t$ act maliciously, do not have greater time and communication costs than ``standard'' executions of the protocol. 
If the protocol designer chooses to use a consensus protocol~$\Prot$ for which redundant executions may be costly, she can often slightly modify the protocol  to alleviate this cost. 
For example, recall that a correct process~$i$ participates in a redundant execution only if it has decided before entering the base protocol. All of our layers ensure that, in this case, all correct processes participating in the protocol propose the same value as~$i$ does. 
The designer can therefore have a correct process that has decided before entering~$\Prot$ simply stop executing~$\Prot$ once it observes a scenario that is inconsistent with all correct processes proposing the same initial value as its own. 
Another useful observation is that in many consensus protocols, including all of the ones quoted in the Introduction, all messages that are sent by correct processes, are sent to everyone. The designer can thus safely modify such a protocol~$\Prot$ in the following manner: Whenever a process~$i$ receives no messages from a process~$j$, it simulates the actions that~$j$ would have performed under~$\Prot$ if it had the same initial value and received the same messages as~$i$ did. Process~$i$ then acts as if it received the messages that the simulated~$j$ would have sent it. The costs of a correct process in an execution of the modified protocol in which $f\le t$ processes are faulty will not exceed those of~$\Prot$ under the same circumstances.

\subsection{Multi-valued Consensus} 
\label{sec:MulValByzCons}
In multi-valued consensus, $|V|\ge3$ and the splitting technique used by $\Ltwo$ and $\Lthree$ for broadcasting values needs to be modified. The resulting technique becomes more cumbersome and less efficient. We design two layers $\Ltwo'$ and~$\Lthree'$ for multi-valued consensus that closely resemble $\Ltwo$ and~$\Lthree$, respectively. The new layers differ from the original ones in two minor ways. One is that the majority computation on line 05 of the original layers is replaced by a plurality computation, which chooses a value that appears most frequently. The other difference is that the 
 broadcasting of values is implemented in a more straightforward manner: In the first and second rounds, processes encode by silence only a common proposal~$\vcom_i$ and decision~$\deccom$, while broadcasting the rest of the values explicitly.
%
A full description of the resulting layers appears in \cref{sec:MV-layer}.
We can show:
\begin{restatable}{theorem}{mvLayers}
	\label{thm:mvLayers}
	Let $k\ge3$, let~$\BB$ be a multi-valued consensus protocol for $n>kt$, and let $L\in \{\Ltwo, \Lthree \}$. Then $\CC'=L'\odot\BB$ is a multi-valued consensus protocol in which 
	\vspace{-2mm}
	\begin{enumerate}
		\setlength\itemsep{-1mm}
		\item In failure-free runs of $\Ltwo'$ (resp.~$\Lthree'$) decisions occur after 2 (resp.~3) rounds and at most $4n(t+1)\log_2|V|$ (resp.~$2n(t+1)\log_2|V|$) bits are communicated, while
		
		\item When failures occur, at most $4n(t+1)\log_2|V|+n^2$ (resp.~$2n(t+1)\log_2|V|+2n^2$) bits are sent and 3 (resp.~4) rounds elapse before reverting to the base protocol. 
		
		\item In the common case, $\CC'$ decides on a plurality value. 
	\end{enumerate}
\end{restatable}

\section{Silent Confirmation Rounds}
\label{sec:SCR}
In the first round of~$\Lone$, no communication occurs in the common case. Similarly, the third round of~$\Ltwo$ is silent in the common case. In~$\Lthree$, both the third  and the fourth round are silent in the common case executions. This is no coincidence. In all of these cases, we use the fact that a process receives no messages to convey relevant information at a small communication cost. 
In this section we formalize the role that such silent rounds play in transmitting information. 

The idea is as follows.
Suppose that we are interested in discovering whether a global property of the correct processes holds. 
In particular, we consider a  ``system milestone'' that is composed of local milestones, one for every process. 
Examples of such a milestone are ``all values are~1'' or ``all correct processes have successfully decided.''
Formally, we define the local milestone of a process~$i$ to be the $i$-local fact~$\varphi_i$, and the system milestone to be~$\Gfact{c}{}\triangleq \bigwedge\limits_{\mbox{\tiny correct}\ i }\varphi_i$.
Information about such milestones can be conveyed in silence as follows. 

\begin{definition}\label{def:silentConfirmationRound}
	For every~$i\in\Proc$, let $\varphi_i$ be an $i$-local fact in the system~$\RP=R({\cal P},\modelt)$.
	Denote $\Gfact{c}{}\triangleq \bigwedge\limits_{\mbox{\tiny correct}\ i }\varphi_i$, and fix some time~$m\ge0$.
	We say that~$\Prot$ implements a \textbf{silent confirmation round for $\Gfact{c}{}$} ($\SCRo{\Gfact{c}{}}$ for short) in round~$m+1$ if in every run $r\in\RP$, each correct~$i\in\Proc$ sends 
	messages to everyone in round~$m+1$ in case~$\varphi_i$ does not hold at time~$m$.
\end{definition}

\begin{theorem}\label{thm:silentConfirmationRound}
	Assume that~$\Prot$  implements an $\SCRo{\Gfact{c}{}}$ in round~$m+1$, and fix a run~$r$ of~$\Prot$. 
	A process~$j$ that receives no messages whatsoever in round~$m+1$ knows at time~$m+1$ that $\Gfact{c}{}$ was true at time~$m$. 
\end{theorem}

\begin{proof}
	Suppose that the assumptions hold and that~$j$ does not receive any round~$m+1$ message in a run~$r\in \RP$.
	We show that~$j$ knows at time~$m+1$ in~$r$ that~$\Gfact{c}{}$ was true at time~$m$.
	Fix some run~$r'\in \RP$ that is $j$-indistinguishable from~$r$ at time~$m+1$.
	By definition, $j$ does not receive any round~$m+1$ message in~$r'$ (otherwise it would distinguish~$r'$ from~$r$).
	The fact that~$j$ receives no round~$m+1$ messages in~$r'$ means in particular that no correct process 
	sent~$j$ any round~$m+1$ message. 
	Since~$\Prot$ implements an  $\SCRo{\Gfact{c}{}}$ in round~$m+1$, we have by  \cref{def:silentConfirmationRound} that~$\varphi_i$ holds at time~$m$ in~$r'$ for all correct~$i\in\Proc$.
	Consequently, $\Gfact{c}{}$ also holds at time~$m$ in~$r'$.
	Since this is true for every run~$r'\in\RP$  s.t.\ $r'\ind{j}{m+1}r$, it immediately follows by \cref{def:know} that~$j$ knows at time~$m+1$ in~$r$ that~$\Gfact{c}{}$ was true at time~$m$.
\end{proof}

As suggested above, silent confirmation rounds are a key ingredient in the efficiency of our layers. 
They allow processes to obtain crucial information quickly and at no message costs in the common case.
In~\cref{fig:bestCase}, by not sending a message if~$v_i=1$ in line~01, $\Lone$ implements an \SCR\ for the fact $\Gfact{c}{}\triangleq \mbox{``all correct processes propose~1"}$.
\cref{thm:silentConfirmationRound} thus guarantees that a correct process can decide and halt in line~03 because it knows that $\Gfact{c}{}$ holds.
Layer~$\Ltwo$ implements an \SCR\ for the fact $\alldec\triangleq$``all correct processes have decided" in round~3 (line~09 in \cref{fig:SanhBinByzCons-protocol}). In the common case execution of~$\Ltwo$ round~3 is completely silent and a process learns that $\alldec$ at time~2. 
The fact that all correct processes have decided implies that there is no need to execute the base protocol to reach consensus. Hence, a  process that knows that everyone has decided can safely halt. 
This is exactly what happens on line 14 of~$\Ltwo$, following a silent round. 
$\Lthree$ uses two silent confirmation rounds, thereby achieving even better communication efficiency.
A round~3 \SCR\ for ``all received a unanimous recommendation" is implemented in line~09 of~\cref{fig:binByzCons-protocol}, which is followed by a round~4 \SCR\ for $\alldec$ in line~11.

Silent confirmation rounds are not restricted to this paper alone. Indeed, the Atomic Commitment protocols in~\cite{GM18} use silent confirmation rounds to gain communication efficiency. 
Moreover, broadcast-based protocols for radio networks use such rounds to overcome possible malicious behavior (see, e.g., \cite{CDGLNN08,GGN09}). 
Silent confirmation rounds can be used to improve solutions to other distributed problems.
For example, a variety of protocols in the literature make use of long silent phases consisting of~$(t+1)$ rounds or more to verify that a specific milestone has been reached (e.g., \cite{GW17,amdur1992message,HH93}).
The time complexity of these protocols in the common case can easily be reduced by employing a silent confirmation round instead.

\subsection{A General View}
\label{sec:Gen}
Both this paper and~\cite{GM18} present  silence-based primitives and use them to aid the design of efficient protocols.
This section provides a broader view on the uses of silence and primitives based on it.
We first give a generalization that gathers such primitives under the same roof and clarifies their relation.
Roughly speaking, we wish to capture the general concept of a protocol purposely not sending messages in order to transmit relevant information. 
We now modify  the definition of a silent confirmation round, by making the set of senders and the set of receivers parameters. 
We proceed as follows. 
\begin{definition}[$S$-$T$ silent broadcast]
\label{def:deliberateSilence}
	Let~$\gamma$ be a synchronous message-passing context in which correct processes follow the protocol, and let  $\RP=R(\Prot,\gamma)$. Fix two sets of processes~$S,T\subseteq\Proc$. For every~$i\in S$, let $\varphi_i$ be an $i$-local fact in the system~$\RP$.
	Denote \mbox{$\Gfact{c}{S}\triangleq \bigwedge\limits_{\mbox{\tiny correct}\ i \in S}\varphi_i$}, and fix some time~$m\ge0$.
	We say that~$\Prot$ implements a \textbf{silent broadcast} 
	\textbf{of~$\Gfact{c}{S}$ from~$S$ to~$T$}  in round~$m+1$, if in every run $r\in\RP$, each correct~$i\in S$ sends messages to each~$j\in T$ during round~$m+1$ in case~$\varphi_i$ does not hold at time~$m$.
\end{definition}

In an analogous manner to \cref{thm:silentConfirmationRound}, an $S$-$T$ silent broadcast guarantees the information transfer property described below.

\begin{theorem}\label{thm:generalSilence}
	Assume that~$\Prot$  implements silent broadcast of~$\Gfact{c}{S}$ from~$S$ to~$T$  in round~$m+1$, and fix a run~$r$ of~$\Prot$.
	A process~$j\in T$ that receives no messages from~$S$ in round~$m+1$ knows at time~$m+1$ that $\Gfact{c}{S}$ was true at time~$m$. 
\end{theorem}
\begin{sketch}
The proof is completely analogous to that of \cref{thm:silentConfirmationRound}, except that \cref{def:deliberateSilence} is used instead of \cref{def:silentConfirmationRound}.
\end{sketch}

Recall that \cref{def:deliberateSilence} explicitly generalizes \cref{def:silentConfirmationRound}. Indeed, a silent confirmation round is an instance of a silent broadcast from $S$ to~$T$, for $S=T=\Proc$.  Since, in addition, $\modelt$ is a synchronous message-passing context, \cref{thm:generalSilence} implies in particular \cref{thm:silentConfirmationRound}.
With appropriate choices for the parameters~$S$ and~$T$, silent broadcast can capture other familiar methods for information transfer using silence. Thus, for example, a {\em null message} from~$i$ to~$j$ in the sense of Lamport~\cite{lamport1984using} is a silent broadcast from~$S=\{i\}$ to~$T=\{j\}$.
Another example is the  \emph{silent choir} from~\cite{GM18} which is a silent broadcast of \mbox{$\psi\triangleq \bigwedge\limits_{\mbox{\tiny correct}\ i \in S}K_i(v_j=1)$} from a carefully defined set~$S$. 
The sender set~$S$ in that case is chosen to be large enough as to ensure that it contains a correct process.
The fact~$\psi$ thus implies that~$K_i(v_j=1)$ holds for at least one correct process and thus~$v_j=1$ must be true.
Although \cite{GM18} assumes a crash failure model, their notion of a silent choir applies equally well to the byzantine model addressed in this paper.
In an analogous fashion, we can use the notion of $S$-$T$ silent broadcast to define silent choirs in the \emph{``Cores and Survivor sets"} model of~\cite{JM03}.
In this model, that considers correlated failures, a \emph{Core} set of process never fails completely, and thus in every run the \emph{Survivor set} contains processes from each \emph{Core}.
Setting~$S$ to be a \emph{Core} set ensures that it contains at least one correct process and therefore has the same effect as a silent choir.


\section{Conclusions}
\label{sec:discussion}
We have offered a novel approach to improving the performance of Byzantine consensus protocols. 
It consists of a modular method that allows the protocol designer considerable control over the tradeoffs between the behavior of the protocols in the common case and in the uncommon cases. 
A layer custom-designed for the common case is prepended to a Byzantine consensus protocol, to provide the best of both worlds. 
Our layer~$\Ltwo$ achieves optimal decision time in the common case, in two rounds. Its communication costs in this case are better by a factor of $\Omega(n)$ compared to all previously known 
protocols that are similarly time-optimal. The layer~$\Lthree$ guarantees the best known common case bit complexity, improving on the previously known protocols by a factor of 24. Finally, for the 
unanimous common case the Layer~$\Lone$ offers a solution that comes essentially at no cost. Even in the uncommon cases, all of our layers add a very small constant number of rounds, 
and  negligible communication costs. 

While Byzantine protocols are designed to withstand extremely complex and challenging scenarios, in most of their computations none of this materializes, and  no failures occur. By diverting the costs 
of handling complex challenges  to the cases in which they actually occur, our layers allow to reap considerable benefits in the normal course of events. These layers can be used to make Byzantine fault-tolerance practically viable. 

Our work is not the first to address efficiency in the common case in a modular fashion. 
Methodological approaches to handling the distinction between the common case and uncommon cases in the shared-memory domain have been provided in influential works by Kogan, Petrank and Timnat \cite{KP12,Timnat15}. 
In synchronous message-passing systems,  our analysis shows that silent confirmation rounds can be used  to shift complexities from 
the common case to uncommon cases. We are certain that other problem domains in distributed systems can benefit from a similar treatment.

An interesting question left open by our investigation concerns the precise upper and lower bounds on the common-case bit complexity  of Byzantine consensus protocols, and the tradeoff between rounds and communication in this case. 
Despite the four decades of extensive research that have been devoted to consensus and Byzantine consensus, we believe that the field will never fail to raise new interesting and practically relevant questions.



\bibliographystyle{plain}
\bibliography{z}

\appendix


\section{Pseudocode for $\Ltwo$ and $\Lthree$}
\label{sec:pseude}
We start the appendix by providing the full pseudocode for the optimizing layers $\Ltwo$ and $\Lthree$ from \cref{sec:prot2,sec:prot3} respectively.

\begin{figure}[h!] 
	\centering{ 
		\fbox{ 
			\begin{minipage}[t]{150mm} 
				\footnotesize 
				\renewcommand{\baselinestretch}{2.5}
				\setcounter{linecounter}{00} 
				\begin{tabbing} 
					aaaaaaaa\=aa\=aaaaaaa\=\kill  
					{\bf time 0} \\					
					$\forall i\in \Proc$:\\
					\line{AA01} \>		$\forall j\in \left\{0,1,2,...,2t\right\}$ {\bf do}\\
					\line{AA02} \>		\>{\bf if} $j\bmod 2=v_i$ {\bf then} be silent \quad \% {\em \gnew{when $v_i=0$, send nothing to Even numbered processes}} \%\\
					\line{AA03} \>		\>{\bf else} send $v_i$ to~$j$\hspace{1.9cm}\quad \% {\em \gnew{when $v_i=1$, send nothing to Odd numbered processes}} \%
					\\[1ex]
					{\bf time 1}\\ 					
					$\forall j\in \left\{0,1,2,...,2t\right\}$:\\
					\line{AA11} \>		$\texttt{values}_j[i]\gets$
					$\left\{
					\begin{tabular}{ l l }
					$j\bmod 2$ & if no message from~$i$ was received\\
					$(j+1)\bmod 2$ & otherwise \\
					\end{tabular}
					\right\}$ \\
					\line{AA12} \>		$rec_j\gets \maj(\texttt{values}_j)$\\
					\line{AA13} \>		$\forall i\in \Proc$ \\
					\line{AA14} \>		\>{\bf if} $i\bmod 2=rec_j$ {\bf then} be silent \quad \% {\em \gnew{when $rec_j=0$, send nothing to Even numbered processes}} \%\\
					\line{AA15} \>		\>{\bf else} send $rec_j$ to~$i$ \hspace{1.8cm}\quad \% {\em \gnew{when $rec_j=1$, send nothing to Odd numbered processes}} \% 
					\\[1ex]{\bf time 2}\\
					$\forall i\in \Proc$:\\
					\line{AA21} \>		\udash{{\bf if} 
						received same recommendation~$rec$ from all~$\left\{0,1,...,2t\right\}$ {\bf then}  $est_i\gets rec$; \dec($est_i$) and be silent}\\[1ex]
					\line{AA22} \>		{\bf else} \quad\hspace{3.9cm}\% {\em \gnew{no identical recommendation}} \% \\
					\line{AA23} \>		\>{\bf if}
					more than $t$ out of~$\left\{0,1,...,2t\right\}$ recommended value $rec$  {\bf then} $est_i\gets rec$\\
					\line{AA24} \>		\> {\bf else} $est_i\gets v_i$\quad\hspace{2.45cm}\% {\em \gnew{no legitimate recommendation}} \% \\
					\line{AA25} \>		\> send `\help' to all
					\\[1ex]{\bf time 3} \textit{and beyond}\\
					$\forall i\in \Proc$:\\
					\line{AA31} \>		\udash{{\bf if} no `\help' message received {\bf then} halt}\\[1ex]
					\line{AA32} \>		{\bf else}\hspace{.45cm} $\decision\gets \mbox{\fbox{$\BB(est_i)$}}$\\
					\line{AA33} \>		\hspace{.95cm} {\bf if} undecided after time~2 {\bf then} \dec $(\decision)$
				\end{tabbing} 
				\normalsize 
			\end{minipage} 
		} 
		\caption{$\Ltwo$ -- A layer for~$i$ that is optimally fast in the common case.
		}
		\label{fig:SanhBinByzCons-protocol} 
	} 
\end{figure}
\begin{figure}[h!] 
	\centering{ 
		\fbox{ 
			\begin{minipage}[t]{150mm} 
				\footnotesize 
				\renewcommand{\baselinestretch}{2.5}
				\setcounter{linecounter}{00} 
				\begin{tabbing} 
					aaaaaaaa\=aa\=aaaaaaa\=\kill  
					{\bf time 0} \\					
					$\forall i\in \Proc$:\\
					\line{AA01} \>		$\forall j\in \left\{0,1,2,...,t\right\}$ {\bf do}\\
					\line{AA02} \>		\>{\bf if} $j\bmod 2=v_i$ {\bf then} be silent \quad \% {\em \gnew{when $v_i=0$, send nothing to the even group}} \%\\
					\line{AA03} \>		\>{\bf else} send $v_i$ to~$j$\hspace{1.9cm}\quad \% {\em \gnew{when $v_i=1$, send nothing to the odd group}} \%
					\\[1ex]
					{\bf time 1}\\ 					
					$\forall j\in \left\{0,1,2,...,t\right\}$:\\
					\line{AA11} \>		$\texttt{values}_j[i]\gets$
					$\left\{
					\begin{tabular}{ l l }
					$j\bmod 2$ & if no message from~$i$ was received\\
					$(j+1)\bmod 2$ & otherwise \\
					\end{tabular}
					\right\}$ \\
					\line{AA12} \>		$rec_j\gets \maj(\texttt{values}_j)$\\
					\line{AA13} \>		$\forall i\in \Proc$ \\
					\line{AA14} \>		\>{\bf if} $i\bmod 2=rec_j$ {\bf then} be silent \quad \% {\em \gnew{when $rec_j=0$, send nothing to even processes}} \%\\
					\line{AA15} \>		\>{\bf else} send $rec_j$ to~$i$ \hspace{1.8cm}\quad \% {\em \gnew{when $rec_j=1$, send nothing to odd processes}} \% 
					\\[1ex]{\bf time 2}\\
					$\forall i\in \Proc$:\\
					\line{AA21} \>		\udash{{\bf if} 
						received a unanimous recommendation~$rec$ from~$\left\{0,1,...,t\right\}$ {\bf then} $est_i\gets rec$}\quad\% {\em \gnew{send nothing}} \% \\[1ex]
					\line{AA22} \>		{\bf else} $est_i\gets v_i$; \ send `\err' to all;\qquad\hspace{3cm} \% {\em \gnew{no legitimate recommendation}} \%
					\\[1ex]{\bf time 3}\\
					$\forall i\in \Proc$:\\
					\line{AA31} \>		\udash{{\bf if} did not receive any `\err' {\bf then} \dec($est_i$)} \qquad\quad\% {\em \gnew{send nothing}} \% \\[1ex]
					\line{AA32} \>		{\bf else} send `\help' to all	
					\\[1ex]{\bf time 4} \textit{and beyond}\\
					$\forall i\in \Proc$:\\
					\line{AA41} \>		\udash{{\bf if} did not receive any `\help' {\bf then} halt}\\[1ex]
					\line{AA42} \>		{\bf else}\hspace{.45cm} $\decision\gets \mbox{\fbox{$\BB(est_i)$}}$\\
					\line{AA43} \>		\hspace{.95cm} {\bf if} undecided after time~3 {\bf then} \dec $(\decision)$
				\end{tabbing} 
				\normalsize 
			\end{minipage} 
		} 
		\caption{$\Lthree$ -- 3 rounds and $nt$ bits in the common case.
		}
		\label{fig:binByzCons-protocol} 
	} \vspace{7ex}
\end{figure}
\newpage

\section{Full Proofs}
\label{sec:proofs}


The proof of \cref{thm:1round} makes use of the following lemma:
\begin{lemma}\label{lem:helping-decision1}
	Fix a run~$r$ of $\CC1=\Lone\odot\BB$ and let~$i$ be a correct process in~$r$.
	If a correct process~$i$ does not decide at time~1, then all the correct processes participate in the $\BB$ phase from time~1 on. 
\end{lemma}
Observe that by line~03, a process that receives at least one `\err' message participates in the base protocol.

\begin{proof}
	Let~$r$ and~$i$ satisfy the assumptions.
	By lines~03 and~04, if~$i$ is undecided at time~1 it received at least $t+1$ `\err' messages in the first round.
	At least one of them must be from a correct process, who sent `\err' to everyone.
	The claim follows.
\end{proof}

\oneround*

\begin{proof}
	We now prove that $\CC1$ is a binary consensus protocol. Fix a run~$r$ of $\CC1$. We show that~$r$ satisfies Decision, Validity and Agreement:\\
		\textbf{Decision} -- Let~$i$ be a correct process in~$r$. If~$i$ decides at time~1 we are done. If it doesn't, then by \cref{lem:helping-decision1} all correct processes participate in the $\BB$ phase. By the Decision property of $\BB$, process~$i$ completes the execution of line~07 and decides on line~08.\\		
		\textbf{Validity} -- Assume that all correct processes propose the same value~$v$ in~$r$. If~$v=1$, then every correct process sends nothing in the first round by line~01 and therefore, a correct process receives at most~$t$ `\err' messages and decides~1 at time~1 (by lines~03 or~04).
		If on the other hand, $v=0$, then since $n>3t$ a correct process receives more than~$2t$ `\err' messages and therefore performs~$est_i\gets 0$ at time~1 by line~06. All correct processes then start the base protocol with a proposal of~0 (line~07) and by its Validity decide~0.\\				
		\textbf{Agreement} -- let~$i,j\in\Proc$ be correct processes in~$r$. If both decide due to the base protocol they have the same decision from its Validity. A correct process that decides at time~1 never changes its mind and can only decide~1 (there is no 0~decisions at this time), therefore, if both $i$ and~$j$ decide at time~1 they decide the same. We are left to show that if (w.l.o.g.) $i$ decides at time~1 and~$j$ decides later using the base protocol, then their decisions are the same.
		Since~$j$ is correct and undecided at time~1, \cref{lem:helping-decision1} states that all correct processes participate in the base protocol by line~07.
		Additionally, since~$i$ decides at time~1 it decides~1. Hence, $i$ received at most~$t$ `\err' messages, which implies that all correct processes received at most~$2t$ `\err' messages, performed $est\gets 1$ on line~05 and entered the base protocol with $est=1$. From Validity of the base protocol this ensures that~$j$ performed $\decision\gets 1$ and thus $j$'s decision on line~08 is~1
	\begin{enumerate}	
		\item In the unanimous common case, all processes propose~1 and no failures occur. Thus, at time~0 all is silent and no messages are sent. Consequently, at time~1 all processes receive no `\err' messages, and so they decide and halt without exchanging any messages.
		
		\item An `\err' message can be implemented using a single bit. In the worst case, every process sends one bit to all others at time~0, incurring a total cost of~$n(n-1)$ bits. After this, all remaining communication is due to the base protocol.
	\end{enumerate}
\end{proof}
The proofs of \cref{thm:2rounds,thm:3rounds} make use of the following lemma:
\begin{lemma}\label{lem:helping-decision}
	Fix a run~$r$ of $\CC2=\Ltwo\odot\BB$ (resp.~$\CC3=\Lthree\odot\BB$).
	If a correct process~$i$ does not decide at time~2 (resp.~3), then all the correct processes participate in the $\BB$ phase from time~3 on (resp.~4 on). 
\end{lemma}
\begin{proof}
	Let~$r$ and~$i$ satisfy the assumptions, and let~$j$ be a correct process in~$r$.
	Line~09 of $\Ltwo$ (resp.\ line~11 of $\Lthree$) implements a \SCR\ for the global fact 
	$\alldec\triangleq$``all correct processes have decided" in round~3 (resp.\ in round~4).
	Suitably, line~14 (resp.\ line~13) dictates that~$j$ halts and does not participate in the base protocol only if~$j$ received no round~3 (resp.\ round~4) messages whatsoever.
	By line~14 (resp.~13) and~\cref{thm:silentConfirmationRound} we have that~$j$ participates in the base protocol unless it knows at time~3 (resp.\ time~4) that $\alldec$ was true at time~2 (resp.~3).
	Since~$i$ does not decide at time~2 (resp.~3) in~$r$, then $\alldec$ is not true at time~2 (resp.~3). By the knowledge property, $j$ does not know that $\alldec$ was true at time~2 (resp.~3), because it is false.
	Consequently, no correct process~$j$ halts at time~3 (resp.~4) in~$r$, and they all  participate in the $\BB$ phase from time~3 (resp.~3) on. 
\end{proof}
\tworounds*
\begin{proof}
	We now prove that $\CC2$ is a binary consensus protocol. Fix a run~$r$ of $\CC2$. We show that~$r$ satisfies Decision, Validity and Agreement:\\
	\textbf{Decision} -- Let~$i$ be a correct process in~$r$. If~$i$ decides at time~2 we are done. If it doesn't, then by \cref{lem:helping-decision} all correct processes participate in the $\BB$ phase. By the Decision property of $\BB$, process $i$ completes the execution of line~15 and decides on line~16.\\
	\textbf{Validity} -- Let~$i$ be a correct process in~$r$ and assume that all correct processes propose the same value~$v$. Recall that, since $n>3t$ by assumption, the correct processes consist of a strict majority. 
	By the pigeonhole principle, at least $t+1$ processes from the greater Sanhedrin are correct. 
	These correct Sanhedrin members follow the protocol on lines~04 and~05 and compute the majority of votes reported to them, which is~$v$, thus, they recommend to all on~$v$ by lines~07 and~08.
	Thereafter, by time~2, every correct process receives at least those $t+1$ recommendations on~$v$ and sets its estimation to~$v$ either by line~09 (in case of a unanimous recommendation), or by line~11.
	If~$j$ decides at time~2, it decides on its estimation~$v$ by line~09, and we are done. Assume it didn't, then by \cref{lem:helping-decision} $j$ and all other correct process participate in the base protocol on line~15.
	As we have shown, the estimation of all correct processes is set to~$v$ on lines~09 and~11. Thus, all correct processes enter the base protocol with a proposal of~$v$.
	From Validity of the base protocol, this ensures that $j$ performs $\decision\gets v$ on line~15 and decides on~$v$ in line~16.\\
	\textbf{Agreement} -- Let~$i$ and~$j$ be correct processes in~$r$. Assume w.l.o.g.\ that~$i$ does decides no later than~$j$.
	If~$i$ does not decide at time~2 then both it and~$j$ participate in the base protocol and decide according to it. In particular, their decisions satisfy Agreement. 
	Let us assume that~$i$ decides at time~2 on~$v$. Specifically, line~09 is the only line in which a correct process decides at time~2.
	A correct process (such as~$i$) decides in line~09 iff it received a unanimous recommendation on~$v$.
	Recall that  every unanimous recommendation includes a report of at least $t+1$ correct processes. It follows that every correct process receives at least $t+1$ recommendations on~$v$ and therefore sets its estimation to~$v$ in lines~09 or~11. 
	Moreover, since no unanimous recommendation on~$u\ne v$ is possible, if~$j$ also decides at time~2, then it decides on $v$ as well, and Agreement holds. 
	If~$j$ does not decide at time~2, then, by~\cref{lem:helping-decision}, all correct processes participate in the $\BB$ phase. And, since all correct processes fixed their estimations to~$v$ at time~2, they all enter the base protocol with~$est=v$. The Validity of the base protocol ensures that $j$ will decide on the value~$v$ in line~16, ensuring Agreement.
\begin{enumerate}		
	\item 
	In a failure-free run at time~0 every process transmits its proposal to half of the Sanhedrin by silence and the other half by messages. Sanhedrin members have one less message to send in half the cases (to themselves), thus at most a total of $n\lceil(2t+1)/2\rceil-\lfloor(2t+1)/2\rfloor\le n(t+1)$ bits are sent during the first round.
	Since no failures occur $\maj (\texttt{values}_j)$ is the same for every Sanhedrin member~$j\in \{0,...,2t\}$ and they all recommend the same value~$v=\maj (\texttt{values})$.
	At time~1, by lines~07 and~08, Sanhedrin members send their recommendations on~$v$ to half of the processes by silence and the other half by messages. Thus, sending at most a total of $(2t+1)\lceil(n-1)/2\rceil\le n(t+1)$ bits in the second round.
	The unanimous recommendation of the second round causes every~$i\in\Proc$ to decide~$v$ at time~2 by line~09, remain silent in round~3 and halt at time~3 by line~14.
	
	\item Again, correct processes send their proposals to the Sanhedrin at a cost of at most $n(t+1)$ bits in the first round, and correct Sanhedrin members send their recommendations to processes with a total cost of at most $n(t+1)$ bits in the second round.
	The difference lays in the third round, when correct processes may not receive a unanimous recommendation and would therefore send \help\ messages (that ca be implemented using a single bit) by line~13. This costs in the worst case $n(n-1)$ bits. After this, all remaining communication is due to the base protocol.
	
	\item 
	In the common case of a failure-free run, all processes transmit their proposals according to protocol at time~0 and a Sanhedrin member calculates their majority at time~1 on lines~04 and~05. The majority value~$v$ is unique and therefore all Sanhedrin members send the same recommendations on~$v$ by lines~07 and~08 at time~1. All processes receive the unanimous recommendation on $v$ by time~2 and therefore decide on it in line~09.
\end{enumerate}
\end{proof}

\threerounds*
\begin{proof}
	We now prove that $\CC3$ is a binary consensus protocol. Fix a run~$r$ of $\CC3$. We show that~$r$ satisfies Decision, Validity and Agreement:\\
	\textbf{Decision} -- Let~$i$ be a correct process in~$r$. If~$i$ decides at time~3 we are done. If it doesn't, then by \cref{lem:helping-decision} all correct processes participate in the base protocol. By the Decision property of the base protocol, process $i$ completes the execution of line~14 and decides on line~15.\\	
	\textbf{Validity} -- Let~$i$ be a correct process in~$r$ and assume that all correct processes propose the same value~$v$. Recall that, since $n>3t$ by assumption, the correct processes consist of a strict majority.
	By the pigeonhole principle, at least one process $j_c\in\left\{0,1,2,...,t\right\}$ is correct. 
	Process~$j_c$ follows the protocol and at time~1 on lines~04 and~05 it computes that the majority of votes as reported to it. Denote this value~$v$. 
	Consequently, by lines~07 and~08 $j_c$ recommends on~$v$ in the second round.
	Thereafter, at time~2 every correct process receives~$j_c$'s recommendation on~$v$ and therefore sets its estimation to~$v$ ($est\gets v$), either in line~09 due to a unanimous recommendation, or in line~10 because its own initial value is~$v$.
	If~$i$ decides at time~3, by line~11 it decides on its estimation, which we have established is~$v$. The only other option for~$i$ to decide is on line~15 by using the base protocol.
	It remains to show that if~$i$ decides using the base protocol, then its decision is also~$v$.
	Assume that~$i$ decides using the base protocol.
	Since~$i$ is a correct process that does not decide at time~3, by~\cref{lem:helping-decision} all correct processes participate in the base protocol.
	As we have shown, the estimation of every correct process is set to~$v$ at time~2 by lines~09 and~10, and so 
	all correct processes enter the base protocol on line~14 at time~4 with the proposal~$v$.
	The Validity of the base protocol ensures that $i$ sets $\decision\gets v$ on line~14 and that $i$ decides~$v$ on line~15. Hence, we are done.\\	
	\textbf{Agreement} -- Let~$i$ and~$j$ be correct processes in~$r$. Assume w.l.o.g.\ that~$i$ decides no later than~$j$. 
	If~$i$ does not decide at time~3 then both it and~$j$ participate in the base protocol and decide according to it. In particular, their decisions satisfy Agreement.
	Let's assume that~$i$ decides at time~3 on a value~$v$.
	Protocol $\CC3$ implements a silent confirmation round for the global fact $\Gfact{c}{}\triangleq$``a unanimous recommendation was received by all correct processes" in round~3 (lines~09 and~10).
	The \SCR\ information transfer guarantees of \cref{thm:silentConfirmationRound} and by line~11 at time~3, imply that $i$ decides at time~3 only if $\Gfact{c}{}$ was true at time~2. In particular,if $i$ decides by line~11, it decides on its estimate value~($est_{i}=v$).
	Recall that every unanimous recommendation includes at least one correct process' recommendation which it recommended to all. It follows that if two correct processes receive unanimous recommendations, then these recommendations are the same.
	Thus, the $\SCRo{\Gfact{c}{}}$ in round~3 informs~$i$ that all correct processes have their estimations set to~$v$.
	If~$j$ also decides at time~3, line~11, then it decides on its estimation~$v$, and Agreement holds.
	If~$j$ does not decide at time~3, then, by~\cref{lem:helping-decision}, all correct processes participate in base protocol. Moreover, since all correct processes have the same estimate~$v$, they all propose~$v$ to the base protocol in line~14.
	Validity of the base protocol guarantees that $\decision\gets v$ in line~14 and~$j$ decides~$v$ by line~15, ensuring Agreement.
\begin{enumerate}		
	\item 
	In a failure-free run at time~0 every process transmits its proposal to half of the Council by silence and to the other half by messages. Council members have one less message to send in half the cases (to themselves), thus at most a total of $n\lceil(t+1)/2\rceil-\lfloor(t+1)/2\rfloor\le n(t+2)/2 -t/2 $ bits are sent during the first round.
	Since no failures occur $\maj (\texttt{values}_j)$ is the same for every council member~$j\in \{0,...,t\}$ and they all recommend the same value~$v=\maj (\texttt{values})$.
	At time~1, lines~07 and~08, council members send their recommendation on~$v$ to half of the processes by silence and to the other half by messages. Thus, sending at most a total of $(t+1)\lceil(n-1)/2\rceil\le n(t+1)/2$ bits in the second round.		
	The unanimous recommendation on~$v$ of the second round causes every~$i\in\Proc$ to set its estimate to~$est_i\gets v$ at time~2, and remain silent in round~3. Thus, in the common case, no message is sent in round~3.
	At time~3, no message is received and in particular no `\err' message, therefore, every process decides on its estimate, remains quiet in round~4 and halts at time~4.
	In conclusion, the total number of messages/bits sent in a failure-free run of~$\CC3$ is at most $n(t+2)/2 -t/2 + n(t+1)/2<n(t+1.5)$.
	
	\item 
	Again, correct processes send their proposals to the council at a cost of at most $n(t+2)/2 -t/2$ bits in the first round, and correct council members send their recommendations to processes with a total cost of at most $n(t+1)/2$ bits in the second round.
	The difference lays in the third and fourth rounds, when correct processes may not receive a unanimous recommendation and would therefore send `\err' messages by line~10 in the third round and \help\ messages by line~12 in the fourth round (each message can be implemented using a single bit). In the worst case, this adds a cost of $2n(n-1)$ bits in rounds~3 and~4. After this, starting at time~4, all remaining communication is due to the base protocol.
	
	\item 
	In the common case of a failure-free run, all processes transmit their proposals according to protocol at time~0.
	A council member calculates the correct majority value~$v$ at time~1 on lines~04 and~05. The majority value~$v$ is unique and therefore all council members send the same recommendation~$v$ at time~1 by lines~07 and~08. All processes receive the unanimous recommendation on~$v$ by time~2 and therefore set their estimation to~$v$ by line~09 and remain silent. In round~3, no messages are sent in the common case, in particular no `\err' messages. Therefore, every process decides in line~11 on its estimate~$v$ which is the majority value.
\end{enumerate}
\end{proof}

\section{Multivalued Layers}
\label{sec:MV-layer}
In multi-valued consensus, $|V|\ge3$ and the splitting technique used by $\Ltwo$ and $\Lthree$ for broadcasting values needs to be modified. The resulting technique becomes more cumbersome and less efficient. We design two layers $\Ltwo'$ and~$\Lthree'$ for multi-valued consensus that closely resemble $\Ltwo$ and~$\Lthree$, respectively. The new layers differ from the original ones in two minor ways. One is that the majority computation on line 05 of the original layers is replaced by a plurality computation, which chooses a value that appears most frequently. The other difference is that the 
broadcasting of values is implemented in a more straightforward manner: In both the first and second rounds, processes encode by silence only a common proposal~$\vcom_i$ and decision~$\deccom$, while broadcasting the rest of the values explicitly. No further changes are needed. \cref{fig:MulCons2,fig:MulCons3} shows~$\Ltwo'$ and~$\Lthree'$.
The function $\plural(\cdot): V^n\rightarrow V$ is defined:
\[
\plural(\textbf{v})\triangleq \mbox{most common~$v$ in \textbf{v} with some arbitrary tie breaker} 
.
\]

\mvLayers*

Proving $\CC2'$ (resp.\ $\CC3'$) is a consensus protocol stems directly from the proof for $\CC2$ in \cref{thm:2rounds} (resp.\ $\CC3$ in\cref{thm:3rounds}). The only minor modification is in Validity, replacing majority with plurality.
However, since when all correct processes propose the same value~$v$ both plurality and majority have the same result.
Therefore, Validity is maintained for the multi-valued as well.
We are thus left only with proving the rest:
\begin{proof}(for $\CC2'$)
\begin{enumerate}
	\item In a failure-free run at time~0 every process transmits its proposal to all Sanhedrin members by messages or silence (a messages can encode any value by at most $\log_2|V|$ bits). Sanhedrin members have one less message to send (to themselves). A message encodes a value by~$\log_2|V|$ bits. Thus a total of at most $(n-1)(2t+1)\log_2|V|$ bits are sent during the first round.
	Since no failures occur $\plural (\texttt{values}_j)$ is the same for every Sanhedrin member~$j\in \{0,...,2t\}$ and they all recommend the same value~$v=\plural (\texttt{values})$.
	At time~1, Sanhedrin members send their recommendations on~$v$ to all processes. Thus, sending at most $(2t+1)(n-1)\log_2|V|$ bits in the second round.
	The unanimous recommendation of the second round causes every~$i\in\Proc$ to decide~$v$ at time~2 by line~07, remain silent in round~3 and halt at time~3 by line~12.
	
	\item Again, correct processes send their proposals to the Sanhedrin at a cost of at most $(n-1)(2t+1)\log_2|V|$ bits in the first round, and correct Sanhedrin members send their recommendations to processes with a cost of at most $(2t+1)(n-1)\log_2|V|$ bits in the second round.
	The difference lays in the third round, when correct processes may not receive a unanimous recommendation and would therefore send \help\ messages (that can be implemented using a single bit) by line~11. This costs in the worst case $n(n-1)$ bits. After this, all remaining communication is due to the base protocol.
	
	\item In the common case of a failure-free run, all processes transmit their proposals according to protocol at time~0 and a Sanhedrin member calculates their plurality at time~1 on lines~03 and~04. The plurality value~$v$ is unique and (a known tie-breaker exists), therefore all Sanhedrin members send the same recommendations on~$v$ by lines~05 and~06 at time~1. All processes receive the unanimous recommendation on $v$ by time~2 and therefore decide on it in line~07.
	
\end{enumerate}
\end{proof}

\begin{proof}(for $\CC3'$)
	\begin{enumerate}
		\item In a failure-free run at time~0 every process sends its proposal the council by messages or silence (a messages can encode any value by at most $\log_2|V|$ bits). Council members have one less message to send (to themselves), thus at most $(n-1)(t+1)\log_2|V|$ bits are sent during the first round.
		Since no failures occur $\plural (\texttt{values}_j)$ is the same for every council member~$j\in \{0,...,t\}$ and they all recommend the same value~$v=\plural (\texttt{values})$.
		At time~1, lines~05 and~06, council members send their recommendation on~$v$ to all the processes. Thus, sending at most $(t+1)(n-1)\log_2|V|$ bits in the second round.		
		The unanimous recommendation on~$v$ of the second round causes every~$i\in\Proc$ to set its estimate to~$est_i\gets v$ at time~2, and remain silent in round~3. Thus, in the common case, no message is sent in round~3.
		At time~3, no message is received and in particular no `\err' message, therefore, every process decides on its estimate, remains quiet in round~4 and halts at time~4.
			
		\item Again, correct processes send their proposals to the council at a cost of at most $(n-1)(t+1)\log_2|V|$ bits in the first round, and correct council members send their recommendations to processes with a total cost of at most $(t+1)(n-1)\log_2|V|$ bits in the second round.
		The difference lays in the third and fourth rounds, when correct processes may not receive a unanimous recommendation and would therefore send `\err' messages by line~08 in the third round and \help\ messages by line~10 in the fourth round (each message can be implemented using a single bit). In the worst case, this adds a cost of $2n(n-1)$ bits in rounds~3 and~4. After this, starting at time~4, all remaining communication is due to the base protocol.
		
		\item In the common case of a failure-free run, all processes transmit their proposals according to protocol at time~0 and a council member calculates their plurality at time~1 on lines~03 and~04. The plurality value~$v$ is unique and (a known tie-breaker exists), therefore all council members send the same recommendations on~$v$ by lines~05 and~06 at time~1. All processes receive the unanimous recommendation on $v$ by time~2 and by line~07 set their estimate to~$v$ and remain silent. Thereafter, at time~3 by line~09 the processes decide on~$v$.
	\end{enumerate}
\end{proof}
\begin{figure}[tbh] 
	\centering{ 
		\fbox{ 
			\begin{minipage}[t]{150mm} 
				\footnotesize 
				\renewcommand{\baselinestretch}{2.5}
				\setcounter{linecounter}{00} 
				\begin{tabbing} 
					aaaaaaaa\=aa\=aaaaaaa\=\kill  
					{\bf time 0} \\					
					$\forall i\in \Proc$:\\
					\line{AA01} \>		{\bf if} $v_i = \vcom_i$ {\bf then} be silent \quad\hspace{.5cm} \% {\em \gnew{$\vcom_i$ -- a common proposal of~$i$}} \% \\
					\line{AA02} \>		{\bf else} send $v_i$ to processes~$\left\{0,1,2,...,2t\right\}$\quad \% \gnew{$v_i\ne \vcom_i$} \%
					\\[1ex]
					{\bf time 1}\\ 					
					$\forall j\in \left\{0,1,2,...,2t\right\}$:\\
					\line{AA11} \>					\>\> $\texttt{values}_j[i]\gets$
					$\left\{
					\begin{tabular}{ l l }
					$\vcom_i$ & if no message from~$i$ was received\\
					proposal received from~$i$ & otherwise \\
					\end{tabular}
					\right\}$ \\
					\line{AA12} \>		$rec_j\gets \plural(\texttt{values}_j)$ \quad \% {\em \gnew{$\deccom$ -- a common plurality result}} \%\\
					\line{AA13} \>		{\bf if} $rec_j=\deccom$ \ {\bf then} be silent \quad\quad \% {\em \gnew{recommendation on $\deccom$}} \% \\
					\line{AA14} \>		{\bf else} send $rec_j$ to all \quad\hspace{1cm}\% {\em \gnew{recommendation on $rec_j\ne\deccom$}} \%
					\\[1ex]{\bf time 2}\\
					$\forall i\in \Proc$:\\
					\line{AA21} \>		\udash{{\bf if} 
						received same recommendation~$rec$ from all~$\left\{0,1,...,2t\right\}$ {\bf then} $est_i\gets rec$; \dec($est_i$) and be silent}\\[1ex]
					\line{AA22} \>		{\bf else} \hspace{1.7cm}\% {\em \gnew{no unanimous recommendation}} \% \\
					\line{AA23} \>		\>{\bf if}
					more than $t$ out of~$\left\{0,1,...,2t\right\}$ recommended value $rec$  {\bf then} $est_i\gets rec$\\
					\line{AA24} \>		\> {\bf else} $est_i\gets v_i$\quad \% {\em \gnew{no legitimate recommendation}} \% \\
					\line{AA25} \>		\> send `\help' to all
					\\[1ex]{\bf time 3} \textit{and beyond}\\
					$\forall i\in \Proc$:\\
					\line{AA31} \>		\udash{{\bf if} no `\help' message received {\bf then} halt}\\[1ex]
					\line{AA32} \>		{\bf else}\hspace{.45cm}$\decision\gets \mbox{\fbox{$\BB(est_i)$}}$\\
					\line{AA33} \>		\hspace{.95cm} {\bf if} undecided after time~2 {\bf then} \dec $(\decision)$
				\end{tabbing} 
				\normalsize 
			\end{minipage} 
		} 
		\caption{$\Ltwo'$ -- The {\em greater Sanhedrin} multi-valued variant.}
		\label{fig:MulCons2} 
	} 
\end{figure}

\begin{figure}[bht] 
	\centering{ 
		\fbox{ 
			\begin{minipage}[t]{150mm} 
				\footnotesize 
				\renewcommand{\baselinestretch}{2.5}
				\setcounter{linecounter}{00}
				\begin{tabbing} 
					aaaaaaaa\=aa\=aaaaaaa\=\kill  
					{\bf time 0} \\					
					$\forall i\in \Proc$:\\
					\line{AA01} \>		{\bf if} $v_i = \vcom_i$ {\bf then} be silent \quad\hspace{.5cm} \% {\em \gnew{$\vcom_i$ -- a common proposal of~$i$}} \%\\
					\line{AA02} \>		{\bf else} send $v_i$ to processes~$\left\{0,1,2,...,t\right\}$\quad \% \gnew{$v_i\ne \vcom_i$} \%
					\\[1ex]
					{\bf time 1}\\ 					
					$\forall j\in \left\{0,1,2,...,t\right\}$:\\
					\line{AA11} \>		$\texttt{values}_j[i]\gets$
					$\left\{
					\begin{tabular}{ l l }
					$\vcom_i$ & if no message from~$i$ was received\\
					proposal received from~$i$ & otherwise \\
					\end{tabular}
					\right\}$ \\
					\line{AA12} \>		$rec_j\gets \plural(\texttt{values}_j)$ \quad \% {\em \gnew{$\deccom$ -- a common plurality result}} \%\\
					\line{AA13} \>		{\bf if} $rec_j=\deccom$ \ {\bf then} be silent \quad\quad \% {\em \gnew{recommendation on $\deccom$}} \% \\
					\line{AA14} \>		{\bf else} send $rec_j$ to all \quad\hspace{1cm} \% {\em \gnew{recommendation on $rec_j\ne\deccom$}} \%
					\\[1ex]{\bf time 2}\\
					$\forall i\in \Proc$:\\
					\line{AA21} \>		\udash{{\bf if} 
						received a unanimous recommendation~$rec$ from~$\left\{0,1,...,t\right\}$ {\bf then} $est_i\gets rec$}\quad\% {\em \gnew{send nothing}} \% \\[1ex]
					\line{AA22} \>		{\bf else} $est_i\gets v_i$; \ send `\err' to all;\qquad\hspace{1.5cm} \% {\em \gnew{no legitimate recommendation}} \%
					\\[1ex]{\bf time 3}\\
					$\forall i\in \Proc$:\\
					\line{AA31} \>		\udash{{\bf if} did not receive any `\err' {\bf then} \dec($est_i$)} \qquad\quad\% {\em \gnew{send nothing}} \% \\[1ex]
					\line{AA32} \>		{\bf else} send `\help' to all	
					\\[1ex]{\bf time 4} \textit{and beyond}\\
					$\forall i\in \Proc$:\\
					\line{AA33} \>		\udash{{\bf if} did not receive any `\help' {\bf then} halt}\\[1ex]
					\line{AA34} \>		{\bf else}\hspace{.45cm}$\decision\gets \mbox{\fbox{$\BB(est_i)$}}$\\
					\line{AA35} \>		\hspace{.95cm} {\bf if} undecided after time~3 {\bf then} \dec $(\decision)$
				\end{tabbing} 
				\normalsize 
			\end{minipage} 
		} 
		\caption{$\Lthree'$ -- The {\em smaller Council} Multivalued variant.}
		\label{fig:MulCons3} 
	} 
\end{figure}

In multi-valued consensus, the decision value domain is commonly defined as~$V\cup\{ \bot \}$, where~$\bot$ is some default value.
A key difference between binary and multivalued consensus is that in the latter it is sometimes impossible to guarantee that processes decide on a correct process' proposal.
More precisely, if~$\frac{n}{|V|}\le t$ then it is not always possible to guarantee that in every execution we decide on some correct proposal.
Consider a failure-free run~$r$ where at most~$t$ processes propose the same value and let the decision value in that run be~$v\in V$, a proposal made by a correct process in~$r$.
Denote by $S_v$ the set of processes that proposed~$v$ in $r$.
We construct the run~$r'$ where all processes have the same initial values as in~$r$ except those in~$S_v$ which start with some other initial value $v'\ne v$. In addition, the processes in~$S_v$ are faulty in~$r'$ ($|S_v|\le t$), and perform the same actions as in~$r$, e.g., proposing~$v$.
For all correct processes the runs are indistinguishable and they must decide on~$v$, a value that no correct process proposed.

Many multivalued Byzantine protocols circumvent this issue by a strong tendency to decide on the default value~$\bot$. In a practical sense, deciding on~$\bot$ usually means a ``no-op" or a ``blank" result.
Our multivalued layers do not produce such an effect.
On the contrary, they allow a designer to improve her solution's time and communication costs while also emulating a fair plurality voting in the common case.

\end{document}